\documentclass[letterpaper, 10 pt, conference]{ieeeconf}  
\usepackage{times}
\usepackage{latexsym}
\usepackage[utf8]{inputenc}
\usepackage{graphicx}
\usepackage[ruled,linesnumbered]{algorithm2e}
\usepackage{graphicx}
\usepackage{amsmath,amsfonts,bm}
\usepackage{float}
\usepackage{bbm}
\DeclareMathOperator*{\argmin}{argmin} 
\usepackage{hyperref}\usepackage{url}
\usepackage{multirow}
\usepackage{subcaption}
\usepackage{caption}
\usepackage{xcolor}
\usepackage{amssymb}  
\usepackage{epstopdf}
\usepackage{mathtools}
\usepackage{dsfont}
\usepackage{tikz}
\usepackage{siunitx}
\usepackage{hyperref}

\usepackage{balance}    
\usepackage{amsthm}
\usepackage[noadjust]{cite} 

\def\sq{\mathbin{{\strut\rule{1.25ex}{1.25ex}}}}

\IEEEoverridecommandlockouts                              
\newtheoremstyle{boldStyle}
  {\topsep}
  {\topsep}
  {\itshape}
  {0pt}
  {\bfseries}
  {.}
  { }
  {\thmname{#1}\thmnumber{ #2}\thmnote{ (#3)}}

\newtheoremstyle{italicStyle}
  {\topsep}
  {\topsep}
  {}
  {0pt}
  {\bfseries}
  {.}
  { }
  {\thmname{#1}\thmnumber{ #2}\thmnote{ (#3)}}

\theoremstyle{boldStyle}

\newtheorem{theorem}{Theorem}

\newtheorem{proposition}{Proposition}
\newtheorem{definition}{Definition}

\theoremstyle{italicStyle}
\newtheorem{assumption}{Assumption}
\newtheorem{remark}{Remark}

\renewenvironment{proof}{{\textbf{Proof:}}}{\hfill$\sq$}

\newcommand{\Smax}{S^{\text{max}}_{\mathcal{X}}}

\title{\LARGE \bf
Control of Unknown Nonlinear Systems with\\ Linear Time-Varying MPC
}

\author{Dimitris Papadimitriou$^{1}$, Ugo Rosolia$^{2}$, and Francesco Borrelli$^{1}$
\thanks{$^{1}$The authors are with the Department of Mechanical Engineering,
        University of California at Berkeley, Berkeley, CA 94701 USA
        {\tt\small dimitri@berkeley.edu}}
\thanks{$^{2}$The author is with the Department of Mechanical and Civil Engineering,
California Institute of Technology, Pasadena, CA 91125 USA}
}
\begin{document}
\maketitle

\begin{abstract}
We present a Model Predictive Control (MPC) strategy for unknown input-affine nonlinear dynamical systems. A non-parametric method is used to estimate the nonlinear dynamics from observed data. 
The estimated nonlinear dynamics  are then linearized over time-varying regions of the state space to construct an Affine Time-Varying (ATV) model. Error bounds arising from the estimation and linearization procedure are computed by using sampling techniques. 
The ATV model and the uncertainty sets are used to design a robust Model Predictive Controller (MPC) which guarantees safety for the unknown system with high probability.
A simple nonlinear example demonstrates the effectiveness of the approach where commonly used estimation and linearization methods fail.
\end{abstract}

\section{Introduction}
Learning the underlying dynamics model of a process has been studied extensively in the traditional system identification literature. Such techniques can be roughly classified into linear \cite{nagumo1967learning,lai1982least} and nonlinear \cite{sjoberg1995non,verdult2002non} methods  applied to time-variant and time-invariant systems. For a comprehensive review of the topic we refer the reader to~\cite{lennart1999system}. Both methods have been integrated in the MPC framework in which system identification strategies are used to estimate the model dynamics~\cite{aswani2013provably,kocijan2004gaussian,rosolia2019learning}. For nonlinear systems a linearization step is usually required to efficiently compute the MPC control action.  Common approaches include successive system linearization \cite{cannon2009successive}, feedback linearization \cite{simon2013nonlinear} and real-time  iteration schemes \cite{diehl2005real}. 

Estimation methods are naturally accompanied by statistical uncertainty. Instability and constraint violation can occur if such uncertainty is not taken into account in the control design. For linear systems, one method to account for the  discrepancy between the actual and estimated dynamics in the control design, is to incorporate the estimation uncertainty into a robust control framework. The authors in~\cite{terzi2018learning, dean2019safely} used a linear regression strategy to identify both a nominal model and the disturbance domain to design a robust controller. In~adaptive MPC strategies~\cite{tanaskovic2014adaptive,tanaskovic2019adaptive,bujarCDC18}, set-membership approaches are used to identify the set of possible parameters and/or the domain of the uncertainty that characterize the system's model. Afterwards, robust MPC strategies for additive~\cite{bemporad1999robust, gonzalez2011online} or parametric~\cite{evans2012robust,fleming2014robust,hanema2017stabilizing,abbas2019tube,hanema2020heterogeneously} uncertainty are used to guarantee robust recursive constraint satisfaction. Another approach to estimate nonlinear systems is via Gaussian Process Regression (GPR)   \cite{kocijan2004gaussian,koller2018learning, klenske2015gaussian, hewing2018cautious}. GPR can be used to identify a nominal model and confidence bounds that may be used to tighten the constraint set over the planning horizon.

We consider nonlinear input affine systems subject to bounded additive uncertainty.
Our contributions can be outlined as follows. We first provide a method to estimate the unknown nonlinear dynamics. Then we outline a novel method to linearize the nonlinear estimated dynamics around an open loop trajectory obtained from the previous MPC time step. In particular, we use information from the codomain of the estimated dynamics function to obtain linearization regions in the domain that satisfy certain criteria. The benefit, compared to commonly used linearization techniques, is that by linearizing using information from the codomain we can control for the linearization error. Furthermore, the control design that incorporates this technique,  is computationally more efficient compared to a piece-wise affine control formulation. 
We design a robust MPC that accounts for the estimation and linearization errors as well as the uncertainty acting on the system. The controller safely steers the system to the goal set while providing probabilistic guarantees for constraint satisfaction. Finally, we illustrate the advantages of our method using a simple and intuitive demonstration.

The paper is organised as follows. Section~\ref{seq:probform} describes the problem formulation. The next section introduces a general MPC formulation with time-varying constraints.  In Sections~\ref{sec:ml} and~\ref{sec:mlin} we present the estimation and linearization techniques respectively. Section~\ref{sec:PR-mpc} presents the control algorithm. Finally, Section~\ref{seq:results} presents a simple illustration that showcases the advantages of our method. 

\section{Problem Formulation}\label{seq:probform}
We are interested in controlling a discrete time dynamical system governed by nonlinear dynamics of the following form  
\begin{equation}
	x_{t+1} = f(x_{t}) + Bu_{t} + w_t\label{system_1},
\end{equation}
with time index $t\in\mathbb{Z}^+$,  state vector $x_t\in\mathcal{X}_t\subseteq\mathbb{R}^{n}$, control input vector $u_t\in\mathcal{U}_t\subseteq\mathbb{R}^{m}$, known control matrix $B\in\mathbb{R}^{n\times m}$ and bounded uncertainty $w_t \in \mathcal{W}\subset\mathbb{R}^n$. 
The dynamics function $f$ is unknown and it is estimated from recorded trajectories. Furthermore, the system is subject to the following state and input constraints
\begin{equation}\label{eq:stateInputConstraints}
    x_t \in \mathcal{X} \text{ and } u_{t}\in \mathcal{U}.
\end{equation}

Our goal is to synthesize a control policy $\pi:\mathcal{X}\rightarrow\mathcal{U}$ which steers the system from a starting state $x_s$ to the target state-input pair $(x_f,u_f)$, while satisfying state and input constraints~\eqref{eq:stateInputConstraints}.
For the rest of the paper we will make the following assumptions.

\begin{assumption}
We are given a feasible trajectory that can drive the system from $x_s$ to $x_f$. 
\label{asum:init-traj}
\end{assumption}

\begin{assumption}
We have access to a data set $\mathcal{D}$ consisting of the following states-action  tuples gathered from past trajectories of the system
\begin{equation*}
    (x^j,x^j_+, u^j)
\end{equation*}
where $x^j_+ = f(x^j)+Bu^j$, for $j=1,\ldots,M$.
\label{asum:data-coll}
\end{assumption}

The first assumption is required for the first time step of our algorithm covered in Section~\ref{sec:PR-mpc}. The second assumption allows us to exploit these trajectories to construct a time-varying approximation of the system dynamics from~\eqref{system_1} as will be shown in Sections~\ref{sec:ml} and ~\ref{sec:mlin}. The dynamics function $f$ is unknown, but we assume that we have access to state measurements.

\section{Time-Varying Model Predictive Control}
This section presents the time-varying MPC formulation to synthesize the control policy $\pi$. More specifically, at time $t$ we solve
\begin{subequations}\label{eq:FTOCP}
    \begin{align}
        \!\!J^*(x_t,T_t) = \!\! \min_{{u_{t|t}}, \ldots, u_{t+T|t}}  ~ &\sum_{k=t}^{t+T_t-1} \!\!\! h(\bar x_{k|t},u_{k|t}) + Q(\bar x_{t+T_t|t}) \notag \\
        \text{s.t.} \quad\quad~ & \bar x_{t|t} - x_t \in \mathcal{E}^{\alpha}_{t|t}, u_{k|t} \in \mathcal{U}_{k|t}  \\
        & \bar x_{k+1|t} = \hat{f}_{k|t} (\bar x_{k|t}) +  Bu_{k|t} \label{eq:nominal}\\
        & \bar x_{k|t} \in \mathcal{X}_{k|t} \ominus \mathcal{E}^{\alpha}_{k|t}, \label{5d}\\
        & \bar x_{t+T_t|t} \in \mathcal{O} \ominus \mathcal{E}^{\alpha}_{t+T_t|t} \\
        &\forall k = {t, \ldots, t+T_t-1},
    \end{align}
\end{subequations}
where  the stage cost satisfies $h(x,u) > 0,~\forall x \in \mathbb{R}^n\setminus\{x_f\},~\forall u\in \mathbb{R}^m\setminus\{u_f\}$
and $h(x_f, u_f)=0$.
In~\eqref{eq:FTOCP} $x_{k|t}$ is the predicted state at time $k$ computed at time $t$. $T_t$ denotes the MPC horizon, $Q$ the terminal convex cost function while $x_t$ is the actual state of the system at time $t$\footnote{The symbols $\oplus$ and $\ominus$ denote the Minkowski sum and Pontryagin difference respectively.}. $\hat{f}_{k|t}$ is the time-varying approximation of the nonlinear  dynamics in~\eqref{system_1} and $\mathcal{X}_{k|t}$ is the set in which this approximation is valid. The sets $\mathcal{E}_{k|t}^\alpha$  represent the uncertainty from the estimation errors and model noise for a confidence level $\alpha$. The expression $\mathcal{X}_{k|t} \ominus \mathcal{E}^{\alpha}_{k|t}$ quantifies how conservative the controller should be in order to complete the task successfully in the presence of the aforementioned uncertainties. This constraint tightening technique is standard in robust MPC  \cite{bemporad1999robust, gossner1997stable}. The subsequent sections detail the construction of the sets $\mathcal{X}_{k|t}$
and $\mathcal{E}^{\alpha}_{k|t}$.

The optimal input sequence obtained from problem~\eqref{eq:FTOCP} $\bar u^*_{t|t},\ldots, \bar u^*_{t+T_t-1|t}$
steers the system from the current state $x_t$ to the goal set $\mathcal{O}$.
Given the optimal solution the MPC policy is
\begin{equation}\label{eq:MPC_policy}
    u_t = \pi(x_t) = u_{t|t}^*.
\end{equation}
At the next time step $t+1$, we solve~\eqref{eq:FTOCP} using the new measured state $x_{t+1}$ and the whole process is repeated until the goal set is reached.

In what follows we show how to construct the prediction model $\hat f_{k|t}$ from historical data. First, we introduce a nonlinear estimator $\hat f$ to identify the system dynamics in~\eqref{system_1}. Afterwards, we propose a sampling-based linearization strategy to approximate the nonlinear estimator with a linear function.
We also compute error bounds and the regions of the state space where the approximation is valid.

\section{Nonlinear Model Estimation}\label{sec:ml}
In this section, we present the nonlinear estimation strategy. In particular, we use the stored data set $\mathcal{D}$ from Assumption~\ref{asum:data-coll} along with non-parametric regression to estimate the nonlinear dynamics function \cite{rosolia2019learning, atkeson1997locally}. The estimate $\hat{f}(x_t)$ at point $x_t$ is given by a linear regression of the points $f(x^j)$ on $x^j$ $, \forall j$. These points are weighted by a chosen kernel function  $K$, which in our case is the Epanechnikov function \cite{epanechnikov1969non}. The estimated function at a generic state $x_t$ is computed by solving the following problem
\begin{equation}
   \hat{a}(x_t), \hat A(x_t) = \underset{ \substack{  A\in\mathbb{R}^{n\times n} \\ a\in\mathbb{R}^n}}{\text{argmin}}~\sum_{j=1}^{M}||f(x^j)\!-\!a\!-\!Ax^j||_2^2K(x_t, x^{j}),
    \label{llr}
\end{equation}
where $\hat A(x_t)$ denotes the dependence of $\hat A$ on the state $x_t$, $Ax^j$ denotes a matrix vector multiplication and the point estimate $\hat f(x_t)$ is given by
\begin{equation}
    \hat{f}(x_t)=\hat{a}(x_t)+\hat{A}(x_t) x_t.
    \label{eq:linear-dyn}
\end{equation}

In the next section we will be using~\eqref{llr} to obtain point estimates of the dynamics on states belonging to previously computed trajectories. Furthermore, we will be defining regions of the state space where these affine system approximations are valid.

Estimating the model from a finite set of data points introduces statistical uncertainty. We are interested in estimating confidence regions for~\eqref{eq:linear-dyn} of the form
\begin{align}
    \mathbb{P}(\hat{f}(x)\in \mathcal{S}_{\alpha}(x))\geq 1-\alpha, \forall x \in \mathcal{X},
    \label{eq:alpha_def}
\end{align}
where $\mathcal{S}_{\alpha}(x)$ denotes the set  we expect $\hat f(x)$ to lie in with probability $1-\alpha$. We quantify this uncertainty by calculating percentile confidence intervals using bootstrap samples from $\mathcal{D}$ \cite{efron1994introduction}. We denote the lower and upper bound of that set with $\mathcal{S}_{\alpha/2}(x)\in\mathbb{R}^n$ and $\mathcal{S}_{1-\alpha/2}(x)\in\mathbb{R}^n$ respectively. The worst case n-ary Cartesian product of uncertainties that occurs within a set $\mathcal{X}$ is denoted with
\begin{equation}
    \Smax:=\left[ \underset{x\in\mathcal{X}}{\text{min}} \mathcal{S}_{\alpha/2}(x),\underset{x\in\mathcal{X}}{\text{max}} \mathcal{S}_{1-\alpha/2}(x)\right],
    \label{est-uncert-max}
\end{equation} 
where the min and max are taken component-wise and the dependence of $S^{\text{max}}_{\mathcal{X}}$ on $\alpha$ is omitted for simplicity. After obtaining   $\hat{f}$ at $x_t$ using (\ref{llr}) the estimated model of the system~is
\begin{equation*}
    x_{t+1}=\hat{a}(x_t)+\hat{A}(x_t) x_t + Bu_t 
    \label{eq:estimated},
\end{equation*}
 where $x_t\in\mathcal{X}_t$ and  $u_t\in\mathcal{U}_t$. It should be noted that although $\hat{f}(x_t)$ from~\eqref{eq:linear-dyn}  is a nonlinear function in $x_t$, it is affine for fixed $\hat a$ and $\hat A$.

  Up until now we have encountered two sources of uncertainty. The first one is the model noise $w_t$ and the second is the uncertainty due to the estimation. Before developing the framework that deals with these uncertainties, we need to convert the estimated dynamics function in a format that will allow us to design a robust MPC. The next section outlines a method that adaptively approximates a nonlinear function with a locally affine one.

\section{Affine Time-Varying Model Approximation} \label{sec:mlin}
When the system model is nonlinear and unknown, the optimal control policy may be approximated after estimating the system dynamics. Common approaches first compute a piecewise affine model estimate which is then used to design a Hybrid MPC. However, these strategies are computationally expensive as the Hybrid MPC is being recast as a Mixed Integer Quadratic Program. An alternative approach to approximate the optimal control policy from~\eqref{eq:FTOCP} is to estimate a nonlinear model and then design a nonlinear MPC problem which is solved using a Real Time Iteration (RTI) scheme or a nonlinear optimization solver. 

This section describes a method that locally linearizes the estimated dynamics function. Furthermore, it defines regions where the difference between the linearized and nonlinear estimated dynamics is bounded. This
linearization strategy allows us to formulate~\eqref{eq:FTOCP} as a convex optimization problem
that can be solved efficiently.

In order to predict $x_{t+1}$ as accurately as possible our goal is to linearize the estimated function $\hat{f}$ in a region around $x_t$. Intuitively, in one dimension,  the larger the slope of $\hat f$ is around $x_t$ the tighter the linearization region around $x_t$ should be so that deviations from $x_t$ do not lead to large deviations in $\hat f(x_t)$. To quantify this, we use a linearization technique that incorporates information about the gradient of the estimated dynamics function. Algorithm 1 proposes a technique which  samples and gradually expands the domain space around a chosen linearization point until the error, as measured in the codomain, surpasses a specified threshold.
 
\IncMargin{1.5em}
\begin{algorithm}[h!] 
	\caption{Local Linear Approximation}\label{algo-2}
	\SetAlgoLined
	\textbf{Input:} Linearization states: $x^{\ell}_{k|t}\in\mathbb{R}^{n}$ for $k=t,
	\ldots,t+T_t$\\
	\textbf{Input:} Sampling step: $\Delta x\in\mathbb{R}$\\
	\textbf{Input:} Maximum linearization error: $\epsilon_{lin}$\\
	\textbf{Output:} Linearization regions $\mathcal{X}_{k|t}$ and estimated dynamics $\hat f_{\mathcal{X}_{k|t}}$\\
	\For{$k=t,\ldots,t+T_t$}{
	Compute $ \hat a  \leftarrow \hat{a}(x^{\ell}_{k|t}),\hat A  \leftarrow \hat A(x^{\ell}_{k|t})$ from~\eqref{llr}\\
	Let $\hat{f}_{\mathcal{X}_{k|t}}(x) \leftarrow\hat{a}+\hat A x\label{fhat_mathcal}$\\
	Set $grid = \{x^{\ell}_{k|t}\}$\\
	Set $step \leftarrow 1$\\
	   \While{$\text{max}_{ x\in grid}|\hat{f}(x)-\hat f_{\mathcal{X}_{k|t}}(x)|\leq \epsilon_{lin}$}{
	   $V^{step}\leftarrow \text{set of all permutations of }$ 
	   $(\pm j\Delta x,\ldots,\pm j\Delta x)\in\mathbb{R}^n$, for $j\in\{0,\ldots,step\}$\\ 
	   $grid  \leftarrow \bigcup_{v\in V^{step}}\{x_{k|t}^{\ell}+ v \}$\\
	   $step  \leftarrow step + 1$
	    }
        $\mathcal{X}_{k|t} \leftarrow Conv(grid) \label{mathcalX}$ }
	Return $\mathcal{X}_{k|t},\hat f_{\mathcal{X}_{k|t}}, k=t,\ldots, t+T_t$ 
\end{algorithm}

Algorithm~\ref{algo-2} determines regions in the domain space for which the linearization of the dynamics is accurate within $\epsilon_{lin}$ tolerance. Intuitively, we construct the linearization regions by incrementally expanding a grid around the linearization points. This grid can be seen as a hypercube with grid points on its edges. 
Once the linearization error in one of the grid points becomes greater than the specified threshold we stop expanding the linearization region. Using information from the codomain of the function allows us to linearize in a more informative manner as we are more conservative in regions where dynamics fluctuate significantly.

By devising an adaptive way to linearize the dynamics we managed to approximate the system in a form amenable for robust model predictive control but we also introduced an additional error term. Given our linearization strategy in Algorithm~\ref{algo-2}, we denote the worst case linearization error within an interval $\mathcal{X}_{k|t}$ as
\begin{equation}
   L^{\text{max}}_{\mathcal{X}_{k|t}}:= [-\epsilon_{lin},\epsilon_{lin}]^n.
    \label{eq:eps-lin-max}
\end{equation}
$L^{\text{max}}_{\mathcal{X}_{k|t}}$ is the n-ary Cartesian power of the worst-case linearization error among the sampling points in the grid.

Having gathered all the individual pieces we are now ready to bound the uncertainty in the estimated dynamics by taking into account the three sources of uncertainty: model noise $w_t$, estimation  uncertainty $S_{\mathcal{X}}^{\text{max}}$ and linearization error $L_{\mathcal{X}}^{\text{max}}$.
\begin{remark}
We underline that it would be possible to guarantee error bounds within the grid, and not just the vertices, by leveraging the Lipschitz properties of the nonlinear estimate.
\end{remark}

\begin{definition}[Cumulative Error Sets]\label{def:cumErr}
Let  $x^*_{t|t}, \ldots, x^*_{t+T_t|t}$ be a state sequence,  $\mathcal{X}_{k|t}$ the linearization regions and $\hat f_{\mathcal{X}_{k|t}}$ the estimated affine dynamics from Algorithm~\ref{algo-2} with $x^{\ell}_{k|t}=x^*_{k|t}$ for $k=t,\ldots,t+T_t$, $S^{\text{max}}_{\mathcal{X}_{k|t}}$ the worst case estimation error from~\eqref{est-uncert-max} and  $L^{\text{max}}_{\mathcal{X}_{k|t}}$ the linearization errors from~\eqref{eq:eps-lin-max}. 
Consider the feasible input sequence $[u_{t|t},\ldots, u_{t+T_t-1|t}]$ and the associated systems
\begin{subequations}
\begin{align}
    x_{k+1|t} &= f(x_{k|t}) +  Bu_{k|t} + w_{k|t} \label{def:predTrue},\\ 
    \bar x_{k+1|t} &= \hat f_{\mathcal{X}_{k|t}}(\bar x_{k|t}) +  Bu_{k|t} \label{def:predEst},
\end{align}    
\end{subequations}
where the true state $x_{k|t}$ and the nominal state $\bar x_{k|t}$ are initialized with $x_{t|t} = \bar{x}_{t|t} = x_{t|t}^*$.
Then we define the cumulative error sets as
\begin{equation}
    \mathcal{E}^{\alpha}_{k+1|t}=\hat{f}_{\mathcal{X}_{k|t}}(\mathcal{E}^{\alpha}_{k|t})\oplus\mathcal{W}\oplus S^{\text{max}}_{\mathcal{X}_{k|t}}\oplus L_{\mathcal{X}_{k|t}}^{\text{max}}
    \label{eq-prop},
\end{equation}
with $\mathcal{E}^{\alpha}_{t|t} = 0$ and a confidence level $\alpha$  in~\eqref{eq:alpha_def}.
\end{definition}

\begin{proposition}\label{theorem1}
Let $\mathcal{D}$ be a data set from Assumption~\ref{asum:data-coll}. Assume that a trajectory of length $T_t$ is given at time step $t$, $x^*_{t|t}, x^*_{t+1|t}, \ldots, x^*_{t+T_t|t}$ . Let $\mathcal{X}_{k|t}$ be the linearization regions from Algorithm~\ref{algo-2} using $x^{\ell}_{k|t}=x^*_{k|t}$ for $k=t,\ldots,t+T_t$  and $\mathcal{E}^{\alpha}_{k|t}$ the cumulative error sets from Definition~\ref{def:cumErr}. If the nominal system~\eqref{def:predEst} satisfies $\bar x_{k|t}\in\mathcal{X}_{k|t} \ominus \mathcal{E}^{\alpha}_{k|t}\;\forall k=t,\ldots, t+T_t$, then for the true unknown system~\eqref{def:predTrue} we have that
\begin{equation*}
    \bar{x}_{k|t} \oplus \mathcal{E}^{\alpha}_{k|t}\in\mathcal{X}_{k|t}, ~\forall k=t,\ldots,t+T_t,
\end{equation*}
with probability $(1-\alpha)^{k-t}$.
\end{proposition}
\begin{proof}
The main idea behind the proof is breaking up the dynamics into a nominal and a noise model in which the latter includes all the uncertainty terms. This strategy is fairly standard in shrinking tube robust MPC strategies~\cite{chisci2001systems}. We will prove the proposition by induction. First we decompose the dynamical system $\forall k \in \{t,\ldots, t+T_t\}$ as follows

\begin{align}
    &\bar{x}_{k+1|t}=\hat{f}_{\mathcal{X}_{k|t}}(\bar{x}_{k|t})+Bu_{k|t} \label{eq-nom}\\
    &e_{k+1|t}=\hat{f}_{\mathcal{X}_{k|t}}( e_{k|t})+\bar{w}_{k|t}\label{eq-noise}\\
    &x_{k+1|t}=\bar{x}_{k+1|t}+e_{k+1|t}\label{eq-condition}
\end{align}
with $e_{k+1|t}\in\mathcal{E}_{k+1|t}$. Equations~\eqref{eq-nom} and~\eqref{eq-noise} refer to the dynamics of the nominal system and the error terms respectively. In~\eqref{eq-noise} $\bar{w}_{k|t}$ now includes the model estimation uncertainty and the linearization error on top of the system disturbance $w_{k|t}$. 
 For $k=t$
     \begin{align}
            &e_{t|t} \in \mathcal{E}_{t|t}\\
            &x_{t|t}=\bar{x}_{t|t} \in \mathcal{X}_{t|t}
    \end{align}
 since at time $k=t$ we assume perfect knowledge of the state, $\mathcal{E}_{t|t}=\{0\}$. Now we assume that at time step $k$
 \begin{align}
            &e_{k|t} \in \mathcal{E}_{k|t}\\
            &x_{k|t}=\bar{x}_{k|t}\oplus\mathcal{E}_{k|t} \in \mathcal{X}_{k|t}.
        \end{align}
 Then at the next time step $k+1$ we have that
\begin{align}
    &\bar{x}_{k+1|t}=\hat{f}_{\mathcal{X}_{k|t}}(\bar{x}_{k|t})+Bu_{k|t} \\
    &e_{k+1|t}=\hat{f}_{\mathcal{X}_{k|t}}( e_{k|t})+\bar{w}_{k|t}
\end{align}
where $\bar{w}_{k|t}\in\{w+s+\ell|w\in \mathcal{W},s\in S^{\text{max}}_{\mathcal{X}_{k|t}},\ell\in L^{\text{max}}_{\mathcal{X}_{k|t}}\}$. More concisely using Minkowski sums we have that $e_{k+1|t}\in \hat{f}_{\mathcal{X}_{k|t}}(\mathcal{E}_{k|t})\oplus \mathcal{W}\oplus S_{\mathcal{X}_{k|t}}^{\text{max}}\oplus L_{\mathcal{X}_{k|t}}^{\text{max}}=\mathcal{E}_{k+1|t}$. By assumption we know that $\bar x_{k+1|t}\in\mathcal{X}_{k+1|t} \ominus \mathcal{E}_{k+1|t}$ and using~\eqref{eq-condition} we obtain that $x_{k+1|t}=\bar{x}_{k+1|t}\oplus \mathcal{E}_{k+1|t}\in\mathcal{X}_{k+1|t}$. Therefore, by induction, we have that
\begin{equation*}
\begin{aligned}
    &\bar{x}_{k+1|t}=\hat{f}_{\mathcal{X}_{k|t}}(\bar{x}_{k|t})+Bu_{k|t}, \forall k\in \{t,\ldots, T_t\} \\
    &e_{k+1|t}=\hat{f}_{\mathcal{X}_{k|t}}( e_{k|t})+\bar{w}_{k|t}, \forall k\in \{t,\ldots, T_t\}
\end{aligned}
\end{equation*}

The above is true for a confidence level $\alpha=0$. However, in reality $s_{k|t}\in\mathcal{S}^{\text{max}}_{\mathcal{X}_{k|t}}$ for $k=t,\ldots,t+T_t$ w.p. $1-\alpha$ for some $\alpha>0$. At each time step $k=t,\ldots,t+T_t$ we know that $s_{k|t}\in\mathcal{S}^{\text{max}}_{\mathcal{X}_{k|t}}$ w.p. $1-\alpha$. Hence for an arbitrary $k\in\{t,\ldots , t+T_t\}$ in order for  $x_{k|t} \in \bar{x}_{k|t} \oplus \mathcal{E}_{k|t}\in\mathcal{X}_{k|t}$ to hold we require that $s_{k'|t}\in\mathcal{S}^{\text{max}}_{\mathcal{X}_{k'|t}},\forall k'=1,\ldots, k$ which happens w.p. $(1-\alpha)^{k-t}$.
\end{proof}

\vspace{2pt}
The sections so far quantified the uncertainty that arises due to model noise, estimation and linearization in a probabilistic manner. We believe this is a realistic strategy for identification of smooth nonlinear dynamical systems. In the next Section~\ref{sec:PR-mpc} we leverage such identification strategy to safely control an underlying unknown system.

\section{Model Predictive Control Design}\label{sec:PR-mpc}
In this section we describe the proposed control strategy which guarantees safety with a desired probability. The control action is computed with Algorithm~\ref{algo-1} that takes as inputs the current state $x_t$, the prediction horizon $T_t$ and the predicted trajectory at the previous time step. First, we use the predicted trajectory to initialize a candidate solution for Problem~\eqref{eq:FTOCP} (line~2). Afterwards, in line~3 we compute the affine approximated dynamics $\hat f_{\mathcal{X}_{k|t}}$, the linearization regions $\mathcal{X}_{k|t}$ and cumulative error sets $\mathcal{E}^{\alpha}_{k|t}$ using the strategies from the previous sections. Finally, we solve the finite time optimal control problem using the approximated dynamics. Notice that Problem~\eqref{eq:FTOCP} may not be recursively feasible and therefore we leverage a shrinking horizon methodology using Assumption~\ref{asum:init-traj}. As we will discuss in Theorem~\ref{theorem2}, this strategy allows us to guarantee safety of Algorithm~\ref{algo-1} in closed-loop with the unknown uncertain system~\eqref{system_1} with high probability.


\begin{algorithm}[h] 
	\SetAlgoLined
	Given  $x_t$, $T_t$, optimal trajectory at the previous time step $\bar x^*_{k|t-1},\;k=t,\ldots,t+ T_t$\\
	Set linearization states $x^{\ell}_{k|t} =  x^*_{k|t-1},\;k=t,\ldots, t+T_t-1$\\
	Compute $ \hat{f}_{\mathcal{X}_{k|t}}$, $\mathcal{X}_{k|t}$ using Algorithm~\ref{algo-2}  and $\mathcal{E}^{\alpha}_{k|t}$ using Definition~\ref{def:cumErr} \\
	\uIf{$J^*_t(x_t,T_{t})$ feasible}{
	    Set $T_{t+1} = T_{t}$\\
	    Let $u^*_{t|t},\ldots,u^*_{t+T_t-1}=\argmin J^*_t(x_t,T_{t})$
	}
	\Else{
	Set $\mathcal{E}^{\alpha}_{k|t}=\mathcal{E}^{\alpha}_{k|t-1}$, $\mathcal{X}_{k|t}=\mathcal{X}_{k|t-1}$, $\hat{f}_{\mathcal{X}_{k|t}}=\hat{f}_{\mathcal{X}_{k|t-1}}$\\
	Solve $J^*_t(x_t,T_{t}-1)$\\
    Set $T_{t+1} = T_{t}-1$\\
    Let $u^*_{t|t},\ldots,u^*_{t+T_t-2}=\argmin J^*_t(x_t,T_{t}-1)$
	}
	Return $T_{t+1}$, $u^*_{k|t},\;k=t,\ldots, t+T_t-1$ 
	\caption{Proposed Strategy}
	\label{algo-1}
\end{algorithm}
It is important to clarify why in our proposed algorithm we used the predicted trajectory of the MPC at the previous time step to obtain a linearization trajectory $x^{\ell}_{k|t}$ (line~2).
We use the fact that in MPC at each time step $t$ we obtain an open loop trajectory from the solution of~\eqref{eq:FTOCP}. Along with this trajectory we also obtain the corresponding input sequence that achieves it. Let $x_t=\bar x_{t|t}$ be the current state. After applying the MPC control input~\eqref{eq:MPC_policy} as obtained from Algorithm~\ref{algo-1}, system~\eqref{system_1} will propagate to state $x_{t+1}=f(x_{t})+Bu_{t}+w_{t}$ at the next time step\footnote{Note that the original trajectory was computed using the estimated dynamics, so naturally it will not be the same as if using the actual dynamics.}. Then by solving~\eqref{eq:FTOCP} again, with the only difference being that the new starting state is $\bar x_{t+1|t+1}=x_{t+1}$, we would naturally expect the points of the new open loop trajectory to lie close to their corresponding points from the previous trajectory. More specifically, we expect the euclidean distance between $\bar x_{t+1|t},\bar x_{t+2|t},\dots,\bar x_{t+T-1|t}$ and $\bar x_{t+1|t+1},\bar x_{t+2|t+1},\dots,\bar x_{t+T-1|t+1}$ to be small.  This allows us to use the previously computed trajectory as the trajectory around which we linearize our estimated dynamics, needed for the MPC at the next time step.

\subsection{Properties}
In this section, we show that the proposed controller is feasible for all $t\in\{0,\ldots,N\}$ with probability $(1-\alpha)^N$. In particular, we leverage Proposition~\ref{theorem1} to show that with probability $(1-\alpha)^N$ Algorithm \ref{algo-1} successfully returns a feasible sequence of input actions at all time instances.

\begin{theorem}\label{theorem2}
Consider the policy~\eqref{eq:MPC_policy} in closed-loop with system~\eqref{system_1}. Let Assumptions~\ref{asum:init-traj}-\ref{asum:data-coll} hold and $N$ be the duration of the task. Then the closed-loop system~\eqref{system_1} and~\eqref{eq:MPC_policy} satisfies state and input constraints with probability $(1-\alpha)^N$ at all times $t\in\{0,\ldots,N\}$.
\end{theorem}
\begin{proof}
As in standard MPC theory we proceed by induction~\cite{borrelli2017predictive}. Assume that at time $t$ a finite time optimal control problem solved by Algorithm~\ref{algo-1} is feasible and let
\begin{equation}
    \begin{aligned}
            &[u_{t|t}^*, \ldots, u_{t+T_t-1|t}^*] \\
            &[\bar x_{t|t}^*, \ldots, \bar x_{t+T_t|t}^*]
    \end{aligned}
\end{equation}
be the optimal input and state sequence. Notice, that if
\begin{equation}\label{eq:conditionFeas}
    x_{t+1} - \bar x_{t+1|t}^* \in \mathcal{E}_{1|t},
\end{equation}
then at the next time step $t+1$, we have that the shifted state and input sequences
\begin{equation}
\begin{aligned}
    &[u_{t+1|t}^*, \ldots, u_{t+T_t-1|t}^*]\\
    &[\bar x_{t+1|t}^*, \ldots, \bar x_{t+T_t|t}^*]
\end{aligned}
\end{equation}
are feasible for problem~\eqref{eq:FTOCP} with prediction horizon $T_t-1$ and for $\mathcal{E}_{k|t+1}=\mathcal{E}_{k|t}$, $\mathcal{X}_{k|t+1}=\mathcal{X}_{k|t}$ and $\hat{f}_{\mathcal{X}_{k|t+1}}=\hat{f}_{\mathcal{X}_{k|t}},~\forall k \in \{t,\ldots, T_t-1\}$. From Proposition~\ref{theorem1} we have that~\eqref{eq:conditionFeas} holds with probability $1-\alpha$, therefore problem $J_t^*(x_{t+1}, T_t-1)$ from Algorithm \ref{algo-1} is feasible with probability $1-\alpha$. At time $t+1$ Algorithm~\ref{algo-1} returns a feasible sequence of input actions with probability $1-\alpha$.\\
Concluding, we have shown that if at time $t$ Algorithm \ref{algo-1} returns a feasible sequence of inputs, then at time $t+1$ Algorithm~\ref{algo-1} is feasible with probability $1-\alpha$. By assumption we have that at time $t=0$ Problem~\eqref{eq:FTOCP} is feasible for $T_0 = T$. Therefore, we conclude by induction that for a control task of length $N$, Algorithm~\ref{algo-1} is feasible and are satisfied for all $t\in \{0,N\}$ with probability $(1-\alpha)^N$.
\end{proof}

\vspace{2pt}
The model predictive control design we outlined above provides a framework to learn and control robustly a nonlinear system. Furthermore, the use of constraint tightening by considering the statistical uncertainty and the linearization error while planning has a self-improving effect on the algorithm. More specifically, it encourages the controller to move with smaller steps when these errors are high and hence gather more informative data points. These data points can then be used online to provide more accurate estimates.

\section{Results}\label{seq:results}
This section compares our method to three commonly used approaches to control nonlinear systems. 
\begin{itemize}
    \item 
The first one is an \textit{unconstrained} MPC that uses locally linear approximations of the dynamics around the linearization trajectory (obtained as explained in Section~\ref{sec:mlin}) without imposing any further tightening constraints, i.e., without constraint~\eqref{5d} in~\eqref{eq:FTOCP}. 
\item The second one is a \textit{linear} MPC where the system dynamics are estimated with a linear model around the current state~$x_t$ throughout the domain. 
\item Finally, we compare our strategy with a \textit{naive} method that linearizes the dynamics around a linearization trajectory but determines a priori the size of the regions in which the linearization is valid. More specifically, lines~8-16 in Algorithm~\ref{algo-2} are no longer needed. The regions $\mathcal{X}_{k|t}$ now have the form 
\begin{equation}\label{eq:naivTol}
\mathcal{X}_{k|t}=\{x\;|\;||x-x_{k|t}^*||_2\leq tol\},
\end{equation}
for some user specified tolerance level~$tol$.
 \end{itemize}
 The simplicity of our example gives us visual insights on why the suggested constraint tightening technique is vital to successfully control the system. First we assume that the true underlying dynamics of the system have the following form $x_{t+1}=\sqrt[3]{x_{t}}+u_t+w_t$, with $x_t\in\mathbb{R}, u_t\in\mathbb{R}$ and $w_t\sim
\psi(\mu,\sigma^2,\tau)$. We use $\psi$ to denote a truncated at $\pm \tau$ normal distribution. The distribution has mean $\mu=0$, standard deviation $\sigma=0.2$ and $\tau=0.05$ which in our example correspond to approximately 5\% uncertainty in the dynamics. Given an initial data set $\mathcal{D}$ of size approximately $M\approx200$ we estimate $\hat{f}$ using non-parametric regression with Epanechnikov kernels. The true function $f$ along with the estimated $\hat{f}$ can be seen in Figure~\ref{fig:estim}. 

\begin{figure}[h!]
\centering
  \includegraphics[trim={0 0 0 0.7cm},scale=0.57]{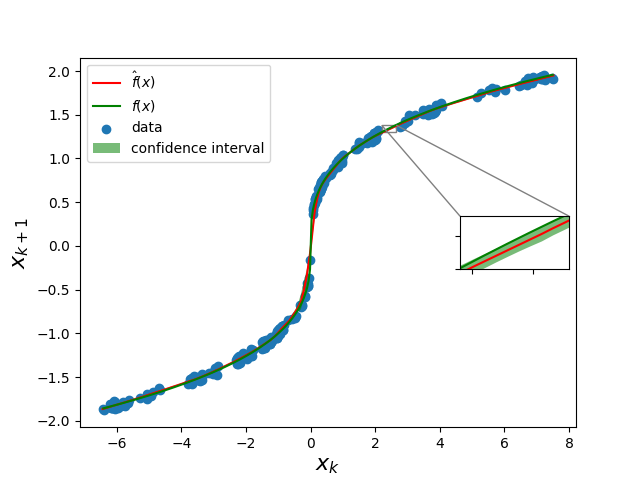}
  \captionof{figure}{True and estimated nonlinear dynamics functions.}
  \label{fig:estim}
\end{figure}

We are interested in driving the system to a terminal state $x_f\in\mathcal{O}$ starting from an initial state $x_s$. At each time step $t$ the MPC is designed with stage and terminal costs $h(x_{k|t}, u_{k|t})=(x_{k|t}-x_f)^TQ(x_{k|t}-x_f)+u_{k|t}^TRu_{k|t}$ and $Q(x_{t+T_t|t})=(x_{t+T_t|t}-x_f)^TQ(x_{t+T_t|t}-x_f)$ respectively.

Furthermore, we set $x_s=4.0$, $x_f=-1.0$, $\mathcal{O} = \{x |\; ||x-x_f||\leq 0.1 \}$, $T_0=6$, $Q=1$, $R=100$, $\alpha=0.05$ and the duration of the task $N=8$. Note that $x_f$ is an equilibrium point for $u_f=0$. Algorithm~\ref{algo-1} requires the trajectory of the previous time step which in the first step is not available. To overcome this in the first time step we use $T_0$ equally spaced points between $x_s$ and $x_f$ around which we perform the linearization.

 \begin{figure}[H]
    \centering
    \includegraphics[trim={0 0 0 0.8cm},scale = 0.53]{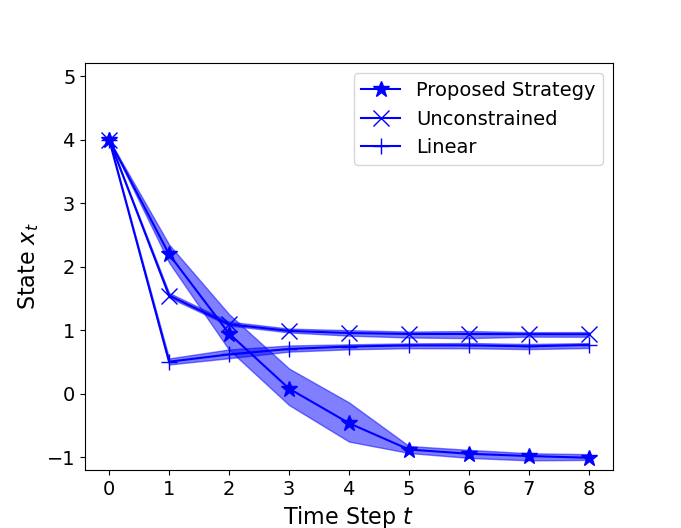}
     \caption{Closed loop trajectories for the \textit{linear}, \textit{unconstrained} and proposed methods.}
    \label{fig:plot_comparison}
\end{figure}
Figure \ref{fig:plot_comparison} shows the closed loop trajectories of our method along with the \textit{linear} and the \textit{unconstrained} ones for $10$ different realizations of the disturbance.
At time $t=0$ all controllers are initialized at $x_s=4.0$. The vertical axis of the plots corresponds to the state $x_t$ and the horizontal axis corresponds to the time step $t=1,\ldots,N$. After $8$ time steps our proposed method reaches the goal set $\mathcal{O}$ while both other methods are trapped around state $x=1.0$. The plotted open loop trajectories in Figure \ref{fig:trajectories_comparison}, along with the dynamics function in Figure~\ref{fig:estim}, provide valuable intuition on why the two other methods fail. The \textit{linear} controller at time step $t=4$ is at $x_{4|4}\approx 0.75$ and plans a trajectory whose first step is state $x_{5|4}\approx0.7$ but due to the modeling mismatch stays around state $0.75$ indefinitely. This modeling mismatch also causes the \textit{unconstrained} controller to fail. More specifically, the controller is overly confident that it can move from state $x_{4|4}\approx 0.9$ to state $x_{8|4}\approx 0.75$ by applying small control inputs and then wrongfully applies most of the control input at the last time step. The fact that the controller initially uses small inputs causes it to also get trapped around state $x_t\approx0.9$. Our proposed strategy uses more accurate linearization of the dynamics, alleviating the effects of modeling mismatch and consistently outperforming the other two methods.
\begin{figure}[h]
    \centering
    \includegraphics[trim={0 0 0 1.1cm},clip,width=9cm,scale = 0.43]{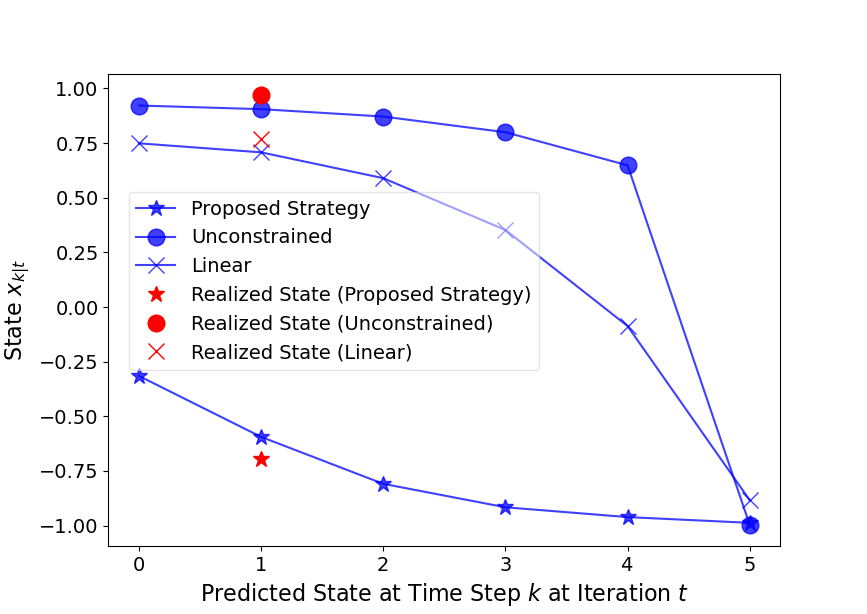}
     \caption{Open loop trajectories at time step $t=4$.}
    \label{fig:trajectories_comparison}
\end{figure}

We also compare our strategy to the \textit{naive} MPC which does not utilize the cumulative error sets for constraint tightening, i.e., $\mathcal{E}^{\alpha}_{k|t}=\emptyset$ in~\eqref{5d}.
Furthermore, in this naive method we compute the linearization regions $\mathcal{X}_{k|t}$ for $\;k=t,\ldots, t+T_t-1$ using~\eqref{eq:naivTol} with $tol$ taking values in $[0.1,1]$ with $0.1$ increments.
We observe that for values of tolerance greater than $0.2$ this method fails while for values smaller than $0.2$ it succeeds. In the successful experiments however, the cumulative cost of the controller is more that four times higher compared to the cost incurred by our method.  Intuitively, when the tolerance level is large, the naive MPC fails as it resembles the \textit{unconstrained} MPC. On the other hand for small tolerance levels the controller is very conservative, leading to successful completion of the task but at a higher cost.

\section{Conclusion}
In conclusion, our work merges elements from machine learning, model linearization and robust control theory to construct a control policy that safely controls a system with probabilistic guarantees. Our method incorporates the model estimation uncertainty and the model linearization error in the formulation mitigating the effects of modeling mismatch while solving a tractable ATV MPC problem. The effectiveness of the proposed strategy is illustrated with a simple  nonlinear example in which commonly  used  estimation  and  linearization  methods  fail.


\section{Acknowledgments}
This work was also sponsored by the Office of Naval
Research grant ONR-N00014-18-1-2833. The views and conclusions contained herein are those of the authors and should not be
interpreted as necessarily representing the official policies or endorsements, either expressed or
implied, of the Office of Naval Research or the US government.

\renewcommand{\baselinestretch}{1.02}
\bibliographystyle{IEEEtran}
\bibliography{IEEEabrv,biblio}

\begin{thebibliography}{10}
\providecommand{\url}[1]{#1}
\csname url@samestyle\endcsname
\providecommand{\newblock}{\relax}
\providecommand{\bibinfo}[2]{#2}
\providecommand{\BIBentrySTDinterwordspacing}{\spaceskip=0pt\relax}
\providecommand{\BIBentryALTinterwordstretchfactor}{4}
\providecommand{\BIBentryALTinterwordspacing}{\spaceskip=\fontdimen2\font plus
\BIBentryALTinterwordstretchfactor\fontdimen3\font minus
  \fontdimen4\font\relax}
\providecommand{\BIBforeignlanguage}[2]{{%
\expandafter\ifx\csname l@#1\endcsname\relax
\typeout{** WARNING: IEEEtran.bst: No hyphenation pattern has been}%
\typeout{** loaded for the language `#1'. Using the pattern for}%
\typeout{** the default language instead.}%
\else
\language=\csname l@#1\endcsname
\fi
#2}}
\providecommand{\BIBdecl}{\relax}
\BIBdecl

\bibitem{nagumo1967learning}
J.-I. Nagumo and A.~Noda, ``A learning method for system identification,''
  \emph{IEEE Transactions on Automatic Control}, vol.~12, no.~3, pp. 282--287,
  1967.

\bibitem{lai1982least}
T.~L. Lai, C.~Z. Wei \emph{et~al.}, ``Least squares estimates in stochastic
  regression models with applications to identification and control of dynamic
  systems,'' \emph{The Annals of Statistics}, vol.~10, no.~1, pp. 154--166,
  1982.

\bibitem{sjoberg1995non}
J.~Sj{\"o}berg, ``Non-linear system identification with neural networks,''
  Ph.D. dissertation, Link{\"o}ping University, 1995.

\bibitem{verdult2002non}
V.~Verdult, ``Non linear system identification: a state-space approach,'' 2002.

\bibitem{lennart1999system}
L.~Lennart, ``System identification: theory for the user,'' \emph{PTR Prentice
  Hall, Upper Saddle River, NJ}, pp. 1--14, 1999.

\bibitem{aswani2013provably}
A.~Aswani, H.~Gonzalez, S.~S. Sastry, and C.~Tomlin, ``Provably safe and robust
  learning-based model predictive control,'' \emph{Automatica}, vol.~49, no.~5,
  pp. 1216--1226, 2013.

\bibitem{kocijan2004gaussian}
J.~Kocijan, R.~Murray-Smith, C.~E. Rasmussen, and A.~Girard, ``Gaussian process
  model based predictive control,'' in \emph{Proceedings of the 2004 American
  Control Conference}, vol.~3.\hskip 1em plus 0.5em minus 0.4em\relax IEEE,
  2004, pp. 2214--2219.

\bibitem{rosolia2019learning}
U.~Rosolia and F.~Borrelli, ``Learning how to autonomously race a car: a
  predictive control approach,'' \emph{IEEE Transactions on Control Systems
  Technology}, 2019.

\bibitem{cannon2009successive}
M.~Cannon, D.~Ng, and B.~Kouvaritakis, ``Successive linearization nmpc for a
  class of stochastic nonlinear systems,'' in \emph{Nonlinear Model Predictive
  Control}.\hskip 1em plus 0.5em minus 0.4em\relax Springer, 2009, pp.
  249--262.

\bibitem{simon2013nonlinear}
D.~Simon, J.~L{\"o}fberg, and T.~Glad, ``Nonlinear model predictive control
  using feedback linearization and local inner convex constraint
  approximations,'' in \emph{2013 European Control Conference (ECC)}.\hskip 1em
  plus 0.5em minus 0.4em\relax IEEE, 2013, pp. 2056--2061.

\bibitem{diehl2005real}
M.~Diehl, H.~G. Bock, and J.~P. Schl{\"o}der, ``A real-time iteration scheme
  for nonlinear optimization in optimal feedback control,'' \emph{SIAM Journal
  on control and optimization}, vol.~43, no.~5, pp. 1714--1736, 2005.

\bibitem{terzi2018learning}
E.~Terzi, L.~Fagiano, M.~Farina, and R.~Scattolini, ``Learning multi-step
  prediction models for receding horizon control,'' in \emph{2018 European
  Control Conference (ECC)}.\hskip 1em plus 0.5em minus 0.4em\relax IEEE, 2018,
  pp. 1335--1340.

\bibitem{dean2019safely}
S.~Dean, S.~Tu, N.~Matni, and B.~Recht, ``Safely learning to control the
  constrained linear quadratic regulator,'' in \emph{2019 American Control
  Conference (ACC)}.\hskip 1em plus 0.5em minus 0.4em\relax IEEE, 2019, pp.
  5582--5588.

\bibitem{tanaskovic2014adaptive}
M.~Tanaskovic, L.~Fagiano, R.~Smith, and M.~Morari, ``Adaptive receding horizon
  control for constrained mimo systems,'' \emph{Automatica}, vol.~50, no.~12,
  pp. 3019--3029, 2014.

\bibitem{tanaskovic2019adaptive}
M.~Tanaskovic, L.~Fagiano, and V.~Gligorovski, ``Adaptive model predictive
  control for linear time varying mimo systems,'' \emph{Automatica}, vol. 105,
  pp. 237--245, 2019.

\bibitem{bujarCDC18}
M.~Bujarbaruah, X.~Zhang, U.~Rosolia, and F.~Borrelli, ``Adaptive {MPC} for
  iterative tasks,'' \emph{2018 IEEE Conference on Decision and Control (CDC)},
  pp. 6322--6327, 2018.

\bibitem{bemporad1999robust}
A.~Bemporad and M.~Morari, ``Robust model predictive control: A survey,'' in
  \emph{Robustness in identification and control}.\hskip 1em plus 0.5em minus
  0.4em\relax Springer, 1999, pp. 207--226.

\bibitem{gonzalez2011online}
R.~Gonzalez, M.~Fiacchini, T.~Alamo, J.~L. Guzm{\'a}n, and F.~Rodr{\'\i}guez,
  ``Online robust tube-based mpc for time-varying systems: A practical
  approach,'' \emph{International Journal of Control}, vol.~84, no.~6, pp.
  1157--1170, 2011.

\bibitem{evans2012robust}
M.~Evans, M.~Cannon, and B.~Kouvaritakis, ``Robust and stochastic linear mpc
  for systems subject to multiplicative uncertainty,'' \emph{IFAC Proceedings
  Volumes}, vol.~45, no.~17, pp. 335--341, 2012.

\bibitem{fleming2014robust}
J.~Fleming, B.~Kouvaritakis, and M.~Cannon, ``Robust tube mpc for linear
  systems with multiplicative uncertainty,'' \emph{IEEE Transactions on
  Automatic Control}, vol.~60, no.~4, pp. 1087--1092, 2014.

\bibitem{hanema2017stabilizing}
J.~Hanema, M.~Lazar, and R.~T{\'o}th, ``Stabilizing tube-based model predictive
  control: Terminal set and cost construction for lpv systems,''
  \emph{Automatica}, vol.~85, pp. 137--144, 2017.

\bibitem{abbas2019tube}
H.~S. Abbas, G.~M{\"a}nnel, C.~H. n{\'e}~Hoffmann, and P.~Rostalski,
  ``Tube-based model predictive control for linear parameter-varying systems
  with bounded rate of parameter variation,'' \emph{Automatica}, vol. 107, pp.
  21--28, 2019.

\bibitem{hanema2020heterogeneously}
J.~Hanema, M.~Lazar, and R.~T{\'o}th, ``Heterogeneously parameterized tube
  model predictive control for lpv systems,'' \emph{Automatica}, vol. 111, p.
  108622, 2020.

\bibitem{koller2018learning}
T.~Koller, F.~Berkenkamp, M.~Turchetta, and A.~Krause, ``Learning-based model
  predictive control for safe exploration,'' in \emph{2018 IEEE Conference on
  Decision and Control (CDC)}.\hskip 1em plus 0.5em minus 0.4em\relax IEEE,
  2018, pp. 6059--6066.

\bibitem{klenske2015gaussian}
E.~D. Klenske, M.~N. Zeilinger, B.~Sch{\"o}lkopf, and P.~Hennig, ``Gaussian
  process-based predictive control for periodic error correction,'' \emph{IEEE
  Transactions on Control Systems Technology}, vol.~24, no.~1, pp. 110--121,
  2015.

\bibitem{hewing2018cautious}
L.~Hewing, A.~Liniger, and M.~N. Zeilinger, ``Cautious nmpc with gaussian
  process dynamics for autonomous miniature race cars,'' in \emph{2018 European
  Control Conference (ECC)}.\hskip 1em plus 0.5em minus 0.4em\relax IEEE, 2018,
  pp. 1341--1348.

\bibitem{gossner1997stable}
J.~R. Gossner, B.~Kouvaritakis, and J.~A. Rossiter, ``Stable generalized
  predictive control with constraints and bounded disturbances,''
  \emph{Automatica}, vol.~33, no.~4, pp. 551--568, 1997.

\bibitem{atkeson1997locally}
C.~G. Atkeson, A.~W. Moore, and S.~Schaal, ``Locally weighted learning for
  control,'' in \emph{Lazy learning}.\hskip 1em plus 0.5em minus 0.4em\relax
  Springer, 1997, pp. 75--113.

\bibitem{epanechnikov1969non}
V.~A. Epanechnikov, ``Non-parametric estimation of a multivariate probability
  density,'' \emph{Theory of Probability \& Its Applications}, vol.~14, no.~1,
  pp. 153--158, 1969.

\bibitem{efron1994introduction}
B.~Efron and R.~J. Tibshirani, \emph{An introduction to the bootstrap}.\hskip
  1em plus 0.5em minus 0.4em\relax CRC press, 1994.

\bibitem{chisci2001systems}
L.~Chisci, J.~A. Rossiter, and G.~Zappa, ``Systems with persistent
  disturbances: predictive control with restricted constraints,''
  \emph{Automatica}, vol.~37, no.~7, pp. 1019--1028, 2001.

\bibitem{borrelli2017predictive}
F.~Borrelli, A.~Bemporad, and M.~Morari, \emph{Predictive control for linear
  and hybrid systems}.\hskip 1em plus 0.5em minus 0.4em\relax Cambridge
  University Press, 2017.

\end{thebibliography}

\end{document}